\documentclass[11pt]{article}
\usepackage{amsmath,amssymb,amsthm,amsfonts}
\usepackage{fullpage}
\usepackage{url}
\usepackage{color}
\usepackage[]{algorithmic}
\usepackage{algorithm}
\usepackage{graphicx}

\DeclareMathAlphabet{\pazocal}{OMS}{zplm}{m}{n}

\newcommand{\Oh}{\mathcal{O}}

\newcommand{\inprod}[1]{\left\langle #1 \right\rangle}

\newtheorem{theorem}{Theorem}

\newtheorem{lemma}{Lemma}

\newcommand{\proofbelow}{3pt}
\newcommand{\afterproof}{\hfill $\blacksquare$ \par \vspace{\proofbelow}}
\renewenvironment{proof}{\noindent\textbf{Proof.}\,}{\afterproof}

\author{Vasileios Nakos\thanks{Harvard University. \texttt{vasileiosnakos@g.harvard.edu}. Supported in part by NSF grant IIS-1447471}} 

\title{Almost Optimal Phaseless Compressed Sensing with Sublinear Decoding Time}

\begin{document}

\maketitle

\begin{abstract}

In the problem of compressive phase retrieval, one wants to recover an approximately $k$-sparse signal $x \in \mathbb{C}^n$, given the magnitudes of the entries of $\Phi x$, where $\Phi \in \mathbb{C}^{m \times n}$. This problem has received a fair amount of attention, with sublinear time algorithms appearing in \cite{cai2014super,pedarsani2014phasecode,yin2015fast}. In this paper we further investigate the direction of sublinear decoding for real signals by giving a recovery scheme under the $\ell_2 / \ell_2$ guarantee, with almost optimal, $\Oh(k \log n  )$, number of measurements. Our result outperforms all previous sublinear-time algorithms in the case of real signals. Moreover, we give a very simple deterministic scheme that recovers all $k$-sparse vectors in $O(k^2)$, using $6k-2$ measurements.

\end{abstract}

\section{Introduction}

Compressive sensing (CS) has emerged during the last decade as a powerful framework for understanding the fundamental limits for signal acquisition and recovery \cite{candes2005decoding,donoho2006compressed}. The basic idea behind CS is that a high-dimensional signal that is sparse in some basis can be recovered from a number of linear measurements, way less than the dimension of the signal. 

In many applications, however, such as X-ray crystallography \cite{sayre1952some}, optical imaging \cite{shechtman2015phase} or astronomy \cite{luke2002optical}, we do not have access to phase information, but rather we can only get the magnitude of each linear measurement. For example, in Coherent Diffractive Imaging \cite{zuo2003atomic}, which is a technique for 2D reconstruction of the image of nanoscale structures, the data we collect is described in terms of absolute counts of photons or electrons; these measurements describe amplitude but loses phases information.

More specifically, in compressive phase retrieval one gets $m$ measurements \[y_i = |\inprod{a_i,x}| ,i=1,\ldots, m,\]
where $a_i$ are the rows of the sensing matrix $A \in \mathbb{C}^{m \times n}$ and $x \in \mathbb{C}^n$. The goal is to approximately recover $x$ from measurements $y_i, i =1, \ldots, m$.

 The problem of recovering a vector $x$ from the magnitudes of its Discrete Fourier Transform has been extensively studied for many decades in the signal processing literature, since its first appearance in 1952. The archetype algorithms that have been devised and are still implemented include the Gerchberg-Saxton algorithm \cite{gerchberg1972practical}, Fienup's algorithm \cite{fienup1987phase} and variants of them. All of the aforementioned algorithms are iterative and follow the approach of alternating projections. Although these algorithms appear to perform well in practice, we have no theoretical guarantees and, moreover, they might not even converge to the initial vector, as they might get stuck at local minima \cite{shechtman2015phase}.

The last years, intensity or phaseless measurements have again been an object of study, this time also under sparsity assumptions \cite{shechtman2015phase}. The first algorithm that guaranteed exact phase retrieval from $\Oh(n)$ Gaussian measurements was PhaseLift \cite{candes2015phase,candes2014solving,candes2013phaselift}, a convex semidefinite programming approach which lifts the vector $x$ to the space of $n \times n$ matrices. However, PhaseLift is computationally very expensive, making it impractical. Non-convex formulations such as Wirtinger flows \cite{candes201flow,chen2015solving} or Truncated Amplitude Flow \cite{wang2016solving} have been developed; the algorithms they use, however, are complicated with many parameters. Very recently, a convex approach called PhaseMax \cite{bahmani2016phase,goldstein2016phasemax} has been suggested and analyzed; the best known analysis of PhaseMax \cite{hand2016elementary} shows that PhaseMax succeeds in finding the underlying vector $x$ under optimal sample complexity.

A decent amount of literature has been devoted to understanding the fundamentals limits of the problem, such as injectivity and uniqueness of solutions \cite{akccakaya2013new}, as well as our ability to design robust algorithms under sparsity assumptions \cite{eldar2014phase,bandeira2013near}. The study of phaseless measurements was initiated in \cite{moravec2007compressive}. In \cite{akccakaya2013new} the authors show that $4k-1$ measurements suffice for recovering all $k$-sparse signals, but they do not suggest any algorithm. In \cite{eldar2014phase} an algorithm for stable recovery is suggested, using $\Oh(k \log (n/k))$ measurements. Using an SDP approach, the papers \cite{ohlsson2011compressive,li2013sparse} show that $\Oh(k^2 \log n)$ measurements suffice to recover exactly $k$-sparse signals using Gaussian matrices using $n^3$ decoding time. To the best of our knowledge, the only papers that achieve sublinear decoding time are \cite{cai2014super,pedarsani2014phasecode,yin2015fast}. The first paper \cite{cai2014super} efficiently reconstructs an exactly $k$-sparse vector $x$ up to a global phase using $\Oh(k)$ measurements in $\Oh(k \log k)$ decoding time. In the other two papers \cite{pedarsani2014phasecode, yin2015fast}, the authors show how to find an arbitrarily large constant fraction of the coordinates of $x$. More specifically, the PhaseCode algorithm of \cite{pedarsani2014phasecode} achieves $\Oh(k)$ measurements for exactly $k$-sparse vectors. The authors in \cite{yin2015fast} show how to robustify PhaseCode so that it also tolerates post-measurements random noise by redesigning the sensing matrix but keeping the main ball coloring algorithm of PhaseCode the same; their scheme needs $\Oh(k \log^3 n)$ measurements in total. All of the three aforementioned algorithms succeed with probability $1 - \Oh(\frac{1}{k})$. There, the authors work in the regime $k = \beta n^{\delta}$ for constants $\beta,\delta$ and moreover each component of $x$ lis in a finite small set. The scheme they use has $\Oh(k \log n)$ measurements and the failure probability is $\frac{1}{k}$.

In this paper we focus on the case of real signals and give a new algorithm called $\textsf{Cphase}_0$, which recovers an approximately $k$-sparse vector from $\Oh(k \log n )$ measurements in time $\Oh(k^{1+o(1)} poly(\log n))$. More specifically, we do not get sign information from the measurements but we want to recover a vector $\hat{x}$ which approximates $x$ up to a universal sign flip. Our algorithm succeeds with probability $1 - o_n(1)$. 
 
Our algorithm departs from the previous techniques and manages to give a sublinear-time algorithm for the case of real signals under the so-called $\ell_2/\ell_2$ guarantee, that achieves $\Oh(k \log n)$ measurements, sublinear decoding time and $o_n(1)$ failure probability. The error guarantee we use has been extensively studied in standard compressed sensing before \cite{gilbert2012approximate, gilbert2013l2} and it is one of the most widely used error guarantees.

\section{Notation}

For a vector $x \in \mathbb{R}^n$ we denote by $x_S$ the vector with the coordinates in $[n]-S$ zeroed out. Also, $\mathrm{head}(k)$ denotes the set of the indices of the $k$ largest in magnitude coordinates of that vector, and $\mathrm{tail}(k) = [n]-\mathrm{head}(k)$. Thus, $x_{\mathrm{tail}(k)}$ is the vector obtained by zeroing out the largest $k$ in magnitude coordinates of $x$. We denote by $|x|$ the vector obtained replacing each entry of $x$ with its absolute value. 

We need the notion of heavy hitters, which is a common concept in the streaming algorithms literature \cite{charikar2002finding}. Intuitively, an $\ell_2$ heavy hitter is a coordinate of the vector $x$ which carries a significant amount of $\ell_2$ mass. Formally, $i$ is an $\epsilon$-$\ell_2$ heavy hitter if $|x_i| \geq \epsilon \|x_{\mathrm{tail}(\lceil \frac{1}{ \epsilon^2 } \rceil)}\|_2$. From now on, when we will refer to a heavy hitter, we will mean an $\ell_2$ heavy hitter.

Our sensing matrix will be denoted by $\Phi$. We also define $y = |\Phi x|$. 
The function $f: x \mapsto |\Phi x|$ is called the sketch of $x$. We note that, in contrast to standard compressed sensing, $f$ is not a linear function.

\section{Our Results}

In this paper we treat the case of signals with real coordinates and show that we can design a scheme with almost optimal measurements and sublinear decoding time, under the $\ell_2/ \ell_2$ guarantee, which is well studied in standard compressed sensing. We now state the main result of this paper.

\begin{theorem}
There exists a distribution $\mathcal{D}$ over matrices $\Phi \in \mathbb{R}^{m \times n} $ associated with a decoding procedure $\textsf{Cphase}_0$, such that $\forall x \in \mathbb{R}^n$, $\hat{x} = \textsf{Cphase}_0(\Phi,|\Phi x|)$ satisfies

\[   \mathbb{P}_{ \Phi \sim \mathcal{D}} [ \mathrm{min}\{\|x - \hat{x}\|_2^2,\|x+ \hat{x} \|_2^2 \}\leq C \|x_{\mathrm{tail}(k)}\|_2^2] > 1 -o_n(1),\]
where $C$ is an absolute constant.
Moreover, the number of measurements of the recovery scheme is $m = \Oh(k \log n)$ and the running time of $\textsf{Cphase}_0$ is $\Oh( k^{1 + \gamma} \mathrm{poly}(\log n)))$, for all constant $\gamma > 0$.
\end{theorem}

In the appendix we also give a very simple deterministic scheme with $6k-2$ measurements that can recover all $k$-sparse vectors in quadratic time in $k$. Previous algorithms for $k$-sparse signals had $\tilde{\Oh}(k)$ time and $\Oh(k)$ measurements, but they were randomized and had a failure probability that was depending on $k$~\cite{cai2014super,pedarsani2014phasecode}.

\section{Robust Compressed Sensing from Phaseless Measurements}

\subsection{Main idea behind the proof}

Because we have access only to the magnitudes of entries of $\Phi x$, a lot of standard algorithms appearing in the compressed sensing literature \cite{gilbert2013l2,gilbert2012approximate,porat2012sublinear} cannot be implemented, since they are based on the linearity of the sketch. However, here our sketch is not linear. Hence, we will make use of sketches that essentially do not use the sign information from the measurements they get: these schemes are the CountSketch \cite{charikar2002finding} and the ExpanderSketch \cite{larsen2016heavy}.\\
In this subsection, we prove a weaker version of Theorem $1$, where the scheme has constant failure probability.  Essentially, we prove the following theorem. We will then show how this algorithm can be modified so that it gives us low failure probability.

\begin{theorem}
There exists a distribution $\mathcal{D}$ over matrices $\Phi \in \mathbb{R}^{m \times n} $ associated with a decoding procedure \textsf{Cphase}, such that $\forall x \in \mathbb{R}^n$, $\hat{x} = Cphase(\Phi,|\Phi x|)$ satisfies

\[   \mathbb{P}_{ \Phi \sim \mathcal{D}} [ \mathrm{min}\{\|x - \hat{x}\|_2^2,\|x+ \hat{x} \|_2^2 \}\leq C \|x_{\mathrm{tail}(k)}\|_2^2] > \frac{2}{3},\]
where $C$ is an absolute consant.
Moreover, the number of measurements of the recovery scheme is $m = \Oh(k \log n )$ and the running time of \textsf{Cphase} is $\Oh( k^{1 + \gamma}poly(\log n)))$, for all constant $\gamma > 0$.
\end{theorem}

In the last subsection of this section, we show how a simple twist can drive the failure probability down to $o_n(1)$.

Our sensing matrix consists of of four different sub-matrices, vertically stacked. The first matrix is $A$, and is used to compute a superset $S$ of the $\frac{1}{\sqrt{10k}}$-heavy hitters of $x$. The second matrix is a standard Count-Sketch matrix \cite{charikar2002finding} and is used to approximate magnitudes of all coordinates in the set $S$, while the matrix $F$ enables us to find the relative signs between all these coordinates. We are going to design $F$ such that every row of it will give us, with some good probability, a pair $\{u,v\}$ with $u,v \in S$. Then we are going to treat this pair as an edge in a graph with vertex set $S$. By the design of $F$ an edge is more likely to exist between two coordinates of the same sign rather than two coordinates of different sign; this will allow us to reduce the problem of finding the relative signs to an instance of the stochastic block model problem, which is well-studied in the random graphs literature. More specifically, we use a measurement of $F$ whenever the support of the corresponding row contains exactly $2$ elements $i,j$ in $S$ and perform a test that indicates if we should add an edge between $i$ and $j$; as mentioned before, there is higher probability of an edge existing between two elements of the same sign rather than elements of opposite sign.

However, there is an obstruction in this stage. The problem is that the existence of coordinates in $S$ with small magnitude can lead to incorrect results. For a measurement, we will call `noise' the mass that participates in that measurement that comes from coordinates outside of $S$. In order to implement the test mentioned in the previous paragraph, it turns out that we should have a low amount of noise. This means, that if a coordinates participates in a measurement on which the noise level is larger than the value of the coordinate, it will not be possible to distinguish this coordinate from the noise and, thus, our test might give us edges between elements of opposite sign more frequently than desired. For this reason, we will use a matrix $E$ which will help us roughly estimate the $\ell_2$ mass of the tail of the vector $x$. Then, we prune $S$ by throwing away all coordinates that are small, in order to guarantee that the average noise level is each meaurement is always smaller than the value of each coordinate in $S$. We will then use a recent result from \cite{abbe2015community} to recognise the two clusters, which correspond to the coordinates with positive signs and to the coordinates with negative signs.

\subsection{Toolkit}

In this subsection, we describe algorithms from the area of streaming algorithms and from the area of random graph theory that will be used for our scheme.
The first result we will make use of is the following theorem from \cite{larsen2016heavy}.

\begin{theorem}{ \cite{larsen2016heavy}}
\label{thm:expandersketch}

For any $n$ and $K<n$ there exists a distribution  $\mathcal{D}$ over matrices $A \in \mathbb{R}^{\Oh(K \log n) \times n} $ and an algorithm \textsf{ExpanderSketch} such that, for every vector $x \in \mathbb{R}^n$, for $y= | \Phi x|$, \textsf{ExpanderSketch}(y) returns a set $S$ which satisfies the following with probability at least $1 - \frac{1}{poly(n)}$:

\begin{itemize}
\item If $|x_i|^2 > \frac{1}{K} \|x_{\mathrm{tail}(K)}\|_2^2$ then $i \in S$
\item $|S| = \Oh(K)$.
\end{itemize}

The running time of \textsf{ExpanderSketch} is $\Oh(K poly( \log n))$.

\end{theorem}

We point out that \cite{larsen2016heavy} is written as working for $y =  \Phi x$, but the algorithm only really needs $|\Phi x|$.\newline

We will also make use of a result  from \cite{abbe2015community}, where the authors elaborate on the well-studied stochastic block model or planted partition problem.

The stochastic block model, $SBM(N,a,b)$, is a random graph ensemble defined as follows:
\begin{itemize}
\item $N$ is the number of vertices in the graph, $V = [N]$ denotes the vertex set.
\item Each vertex $v \in V$ is assigned independently a hidden (or planted) label $R$ or $B$, equally with probability $\frac{1}{2}$.
\item Each (unordered) pair of nodes $(u, v) \in  V \times V$ is connected independently with probability $a \frac{\log N}{N}$ if have the same label, and with probability $b \frac{\log N}{N}$ if they have different label.
\end{itemize}
The above gives a distribution on graphs with $N$ vertices. Note that $G \sim SBM(N,a,b)$ denotes a random graph drawn under this model, without the hidden (or planted) clusters (i.e., the labels ) revealed. The goal is to recover these labels by observing only the graph. The following theorem holds.

\begin{theorem}{ \cite{abbe2015community}}
\label{thm:signs}

There exists an algorithm called \textsf{Degree Profiling} such that for any instance generated by the stochastic block model, it finds the clusters $(R,B)$ with probability at least $ 1 - o_N(1)$, whenever $ \sqrt{a} - \sqrt{b} \geq 2$. The algorithms runs in $\Oh(N^{1+ \gamma})$ time, for all constant $\gamma > 0$.

\end{theorem}

We will abuse notation slightly and assume that the algorithm receives as input a set of edges $(u,v)$, with $u,v \in [N]$ and $a$ is much bigger than $b$. The reader can verify that the Degree Profiling succeeds with probability $1 - o_N(1)$ and takes as input $\Omega( N \log N)$ random edges.

Given a set $S$, we are going to reduce the problem of finding the relative signs of the set $\{x_i \}_{i \in S}$ to an instance of the stochastic block model, which is well-studied in the random graphs literature. We will then run the \textsf{Degree Profiling} algorithm just described. The $R$ cluster is going to be the set of coordinates $x_i \geq 0 0$ and the $B$ cluster is going to be the set of coordinates $x_i < 0$. We will discuss details of this model shortly.

\subsection{Construction of the sensing matrix}

We proceed by defining the matrices we are going to make use of. All the constructions are randomized and hence, as mentioned before, induce a distribution over matrices $\Phi$. Before defining $\Phi$, we define the following random matrices:

\begin{itemize}
\item $D \in \mathbb{R}^{n \times n}$ is a diagonal matrix with random signs, that is $\mathbb{P}[D_{ii} = 1] = \mathbb{P}[D_{ii}=-1] = \frac{1}{2}$. Moreover, $D_{ij} = 0$ for $i \neq j$.
\item $A$ is the matrix guaranteed by $Theorem~2$ for $K=10k$.
\item $B \in \mathbb{R}^{ \Oh(k \log n) \times n}$ is the Count-Sketch from \cite{charikar2002finding}. 
\item $E$ is a $\Oh(k \log n) \times n$ matrix, which consists of $\Oh( \log n)$ submatrices $E_1, \ldots , E_{\Oh( \log n)}$, each one having $\Omega(k)$ rows. In each such submatrix, each element of the matrix equals $0$ with probability $1 - \frac{1}{k}$, $+1$ with probability $\frac{1}{2k}$ and $-1$ with probability $\frac{1}{2k}$. In other words, each element is non-zero with probability $\frac{1}{k}$ and if this is the case, it is equally likely to be one of $+1,-1$.
\item $F$ consists of $\lceil \log (5k) \rceil$ submatrices, $F_1, F_{2^1},\ldots, F_{2^l}, \ldots, F_{2^{\lceil \log(5k)} \rceil}$. Each entry of $F_{2^l}$ is non-zero with probability $\frac{1}{C_0 2^{l} (\log(5k) - l+2)^2}$, where $C_0$ is a large enough constant. Each non-zero entry is equally likely to be $+1$ or $-1$. Every matrix $F_{2^l}$ has $\Oh(l 2^l (\log(5k) - l +2)^4)$ rows, where the constant inside the big-Oh depends on $C_0$. 
\end{itemize}

Now let $\Phi'$ be the vertical concatenation of $A,B,E,F$ and set $\Phi = \Phi' D$. 
For convenience, we will slightly abuse notation and when referring to the row $q$ of some of the aforementioned matrices we will say that it refers to measurement $y_q$. The measurement we are referring to will be clear from context.

Let us look now at the description of the \textsf{Cphase} algorithm. The constant $c_1$ and the parameter $\Delta$ will be chosen later.
The variable $L$ serves as an approximation of $\frac{1}{k}\|x_{tail(k)}\|_2^2$. It is upper-bounded by $\Oh(\frac{1}{k}\|x_{\mathrm{tail}(k)}\|_2^2)$, but it is also $\Omega(\frac{1}{k}\|x_{tail(\Omega(k))}\|_2^2)$. This guarantee for $L$ is crucial for two things. Since $L$ is used to filter out coordinates from $S$ that are not large enough, the lower bound ensures that these coordinates will indeed not be present in $S$. On the other hand, it is guaranteed by the upper bound on $L$ that no large enough coordinate will be thrown away, something which guarantees that there is no significant loss of information for our sparse recovery algorithm.\\

\begin{figure}
\begin{center}
\fbox{                                                                                                         
{\footnotesize                                                                                                 
\parbox{7.in} {  
\underline{Algorithm \textsc{Signs}$(S)$}:\\                                                            
\vspace{-.30in}\begin{enumerate}                                                                               
\addtolength{\itemsep}{-0.5mm}
\medskip

\item $ V \leftarrow S, E \leftarrow \emptyset$
\item $l^* \leftarrow \mathrm{argmin}_l \{ l: |S| \leq 2^l\}$.
\item for every row $q$ of $F_{2^{l^*}}$
		\item \qquad if $|supp(q) \cap \ S| =2 $
			
			\item \qquad \qquad  $\{u,v\}  \leftarrow supp(q) \cap S$
			\item \qquad \qquad if  $\sigma_{qu} = \sigma_{qv}$
				\item \qquad \qquad \qquad if  $| y_q - | |\hat{x}_u| +  |\hat{x}_v||| < | y_q - | |\hat{x}_u| - |\hat{x}_v|| |$
					\item \qquad \qquad \qquad \qquad  $ E \leftarrow E \cup \{u, v\}$
				
			\item \qquad \qquad else 
					\item \qquad \qquad \qquad if  $| y_q - | |\hat{x}_u| +  |\hat{x}_v||| > | y_q - | |\hat{x}_u| - |\hat{x}_v|| | $   
					\item \qquad \qquad \qquad \qquad  $ E \leftarrow E \cup \{u, v\}$

\item Degree-Profiling((V,E))

\end{enumerate}                                                                                                
}}}
\end{center}
\caption{Signs Algorithm}\label{fig:signs-alg} 
\end{figure}


\begin{figure}
\begin{center}
\fbox{                                                                                                         
{\footnotesize                                                                                                 
\parbox{7.in} {  
\underline{Algorithm \textsc{Prune}$(S,L)$}:\\                                                            
\vspace{-.30in}\begin{enumerate}                                                                               
\addtolength{\itemsep}{-0.5mm}
\medskip
\item $\{z_i\}_{i \in [|S|] } \leftarrow \{|\hat{x}_i|\}_{i \in S}$
\item Sort all $z_i$ in decreasing order.
\item Find maximum $m \in [|S|]$ such that $|z_m|_2^2 > \frac{k \cdot L}{C_0 \cdot 2^{l_0} (\log(5 k)- l_0 +2)^2}$ (where $l_0$ is such that $ 2^{l_0} < m \leq 2^{l_0+1} )$

\item $T \leftarrow \emptyset$
\item for  $ i \in S$
	\item \qquad if $|\hat{x}_i| \geq z_m$ 
			\item \qquad \qquad $T \leftarrow T \cup \{i\}$
\item return $T$

\end{enumerate}                                                                                                
}}}
\end{center}
 
\caption{\textsf{Prune} Algorithm}\label{fig:prune-alg} 
\end{figure}



\begin{figure}
\begin{center}
\fbox{                                                                                                         
{\footnotesize                                                                                                 
\parbox{7.in} {                                                                                             
\underline{Algorithm \textsc{Cphase}$(\Phi,y)$}:\\                                                            
\vspace{-.30in}\begin{enumerate}                                                                               
\addtolength{\itemsep}{-0.5mm}
\medskip
\item $S_0 \leftarrow ExpanderSketch(y)$
\item $\forall i \in S_0$, $\hat{x}_i \leftarrow CountSketch(i)$
\item $S_1 \leftarrow$  top $5k$ in magnitude coordinates of $S_0$
\item \textbf{for} $l=1$ to $\Oh(\log n)$
\item  \qquad $\Delta \leftarrow 0$
\item  \qquad \textbf{if} $supp(q) \cap S = \emptyset$
\item \qquad \qquad$W_{\Delta} \leftarrow y_q^2$
\item \qquad \qquad $\Delta \leftarrow \Delta +1$

\item \qquad $W \leftarrow \frac{1}{\Delta} \sum_{i=1}^{\Delta} W_i$	
\item \qquad $L_l \leftarrow c_1 \cdot  W$
\item $L \leftarrow median_{l=1}^{\Oh( \log n)} \{ L_l \}$
\item $ S_2 \leftarrow Prune(S_1, L)$
\item $\Sigma_S \leftarrow Signs(S_2)$
\item For each $i$ in $S_2$, $\hat{x}_i \leftarrow \Sigma_S(i) |\hat{x}_{ i}|$ 
\item Output $\hat{x}$

\end{enumerate}                                                                                                
}}}
\end{center} 

\caption{\textsc{Cphase} Algorithm}\label{fig:cphase-alg} 
\end{figure}

We now compute the number of measurements used by our recovery scheme.

\begin{lemma}
\label{lem:numberofrows}
The total number of rows of $\Phi$ is $\Oh(k \log n)$.
\end{lemma}

\begin{proof}
We only prove that $F$ has $\Oh (k \log k)$ rows, as the number of rows of the other matrices are obvious to calculate. Define $i= \log(5k) - l +2$ and observe that the total number of rows equals 
 \[ \sum_{i=1}^{\log(5k)} ( \log (5k) - i - 2) 2^{\log(5k) -i-1 } \cdot i^4 \leq \Oh( k \log k \sum_{i = 1}^{\infty} \frac{i^4}{2^i}) = \Oh( k \log k),\]

and the proof follows.
\end{proof}

Since we are multiplying $x$ with a random sign matrix $D$ at the beginning, we get a vector $x' = D x$. Observe that the magnitude of each coordinate remains the same, while the sign of its coordinate is now uniformly at random. Given also an approximation of $x'$, we can find an approximation of $x$ by just computing $x = D x'$ in time $supp(x')$. This is a reduction to the case where $sign(x_i)$ are uniformly at random, while $|x_i|$ are adversarially chosen. 

The first step of the decoding algorithm is to run the decoding algorithm from \cite{larsen2016heavy} to get a set $S_0$ containing all $\frac{1}{\sqrt{10k}}$-heavy hitters of $x$, as promised by Theorem~\ref{thm:expandersketch}. Then, we find estimates of all coordinates in $S_0$ using the CountSketch. 

\begin{lemma}{ \cite{charikar2002finding} }
\label{lem:countsketch}
With probability $1 - \frac{1}{poly(n)}$, for any $i \in S_0$, we have that $||\hat{x}_i | - |x_i| |^2 \leq \frac{1}{10k} \|x_{\mathrm{tail}(k)}\|_2^2$.
\end{lemma}



The next step is to keep the largest $5k$ coordinates of $x$ and obtain the set $S_1$. What we prove now is that this operation does not introduce a lot of error and approximating coordinates in $x$ still suffices for our approximation guarantee. For reasons that will be clear later, we would like $\| x_{\bar{S_1}}\|_2$ not to be much bigger that $\|x_{tail(k)}\|_2$, nor much less from $\|x_{\mathrm{tail}(10 k)}\|_2$. The following lemma gives this type of guarantee.

\begin{lemma}
\label{lem:estimation}
With probability $1 - \frac{1}{poly(n)}$, it holds that $1.2\|x_{\mathrm{tail}(k)}\|_2^2 \geq \|x_{\bar{S_1}}\|_2^2 \geq \|x_{\mathrm{tail}(10 k)}\|_2^2$.
\end{lemma}

\begin{proof}

The left-hand side follows by Theorem 3.1 in \cite{price11}.
For the right-hand side, observe that 
\[ \|x_{\bar{S}_1} \|_2^2 \geq \sum_{j \notin \mathrm{head}(k) \cup S_1} x_j^2 \geq \|x_{\mathrm{tail}(10k)}\|_2^2\]

\end{proof}

Our goal now is to find the relative signs inside $S_1$. The idea behind this estimation is to build a graph, where there is an edge between two coordinates $i,j \in S_1$ if there is a row  $q$ of $F$ such that $supp(q) \cap S_1 =\{i, j \}$. This means that these are the only two coordinates in $S$ that contribute to the measurement defined by this row, and hence we hope to extract some information about the relative signs of $i$ and $j$. However, the presence of coordinates of very small magnitude in the set $S$ might lead to wrong estimations, due to the noise coming from coordinates not in $S_1$. This is why we would like to throw away from $S_1$ all $i$ such that $|x_i|$ is `small': so that both $|x_i|$ and $|x_j|$ are above the (average) `noise' level. Here, the noise is the $\ell_2$ mass coming into a measurement from coordinates not in $S_1$. This implies that we need to prune the set $S_1$ to get $S_2$.

In order to implement this pruning, we need to find a threshold that defines if a coordinate is heavy or not. Ideally, we would like this threshold to be multiplicative in $\frac{1}{\sqrt{k}}\|x_{\mathrm{tail}(k)}\|_2$. Unfortunately, this type of guarantee is not possible, but we can find a threshold that still can make our algorithm go through. To compute this threshold we implement the following procedure using the matrices $E_l$: for each matrix, we keep any row that has empty intersection with every element in $S_1$, then we average these measurements and normalize them by a suitable constant $c_1$ to get $L_l$; we then take the median of $L_l$ to get $L$. This is implemented in lines $4-11$ of the \textsf{CPhase} algorithm. The following lemma holds.

\begin{lemma}
\label{lem:cheby}
With probability $1 - \frac{1}{poly(n)}$, we have that $ \frac{1}{k} \|x_{\mathrm{tail}(k)}\|_2^2 \geq L \geq C' \frac{1}{k}\|x_{\mathrm{tail}(10k)}\|_2^2$, where $C'$ is an absolute constant.
\end{lemma}
\begin{proof}
 We focus on a specific matrix $E_l$ and show that the desired inequalities hold with constant probability. Taking the median of all $\Oh( \log n)$ values ensures high probability. 

The support of some row of $E_l$ with constant probability has no intersection with $S_1$. After considering $\Omega( k )$ rows, a constant fraction of the rows will have support not intersecting $S_1$ by a Chernoff bound with probability at least $1 - e^{-\Omega(k)}$ . Consider such rows $r_1,..,r_\Delta$ where $\Delta =  C_1 k$, where $C_1$ is a constant to be chosen later, and define \[W_{l,q} = | \sum_{i=0}^n \delta_{qi} \sigma_{qi} x_i|^2\] for the measurement corresponding to row $r_q$. Since we have conditioned on $S_1$ having empty intersection with any such row, then $\delta_{qj} = 0$ for $ j \in S_1$. Let $\mathcal{E}_q$ be the event that $\delta_{qj} = 0$ for all $j \in S_1$. \\

 We get that \[\mathbb{E} [W_{l,q} | \mathcal{E}_q] = \frac{ 1}{ k} \|x_{\bar{S_1}}\|_2^2\] and 
 \[\mathbb{E} [ (\frac{1}{\Delta} (W_{l,1} + W_{l,2} +.. +W_{l,\Delta}))^2| \mathcal{E} ] = \Delta^{-2}(\frac{1}{ k}\  \sum_{i \in [n] - S} x_i^4 + \frac{1}{k^2}  \sum_{i \neq j, i,j \in [n] - S_1} x_i^2 x_j^2) \leq \Delta^{-2} \frac{1}{ k} \|x_{\bar{S_1}}\|_2^4 \]
 
Moreover, \[ \mathbb{P} [ |\frac{1}{\Delta}[ W_{l,1} + W_{l,2} +.. + W_{l,\Delta}  - \frac{1}{k} \|x_{\bar{S_1}}\|_2^2 |\geq 0.1\frac{1}{k}\|x_{\bar{S_1}}\|_2^2| \mathcal{E}] \leq \frac{\Delta^{-2} \frac{1}{ k} \|x_{\bar{S_1}}\|_2^4}{ 0.1^2 \frac{1}{ k^2}\|x_{\bar{S_1}}\|_2^4} \leq \frac{1}{200},\]for big enough choice of $\Delta = C_1 k$. This means that using each matrix $E_l$ we can find a value $W$ which is within $ (1 \pm 0.1)\frac{1}{k} \|x_{\bar{S}_1} \|_2^2$ with constant probability. This, along with $Lemma~3$ implies that with  $L_l$ will satisfy the guarantee of the lemma with constant probability. Since we are taking the median of all $L_l$ the aforementioned guarantee will hold with probability $1- \frac{1}{poly(n)}$. 
\end{proof}

After the estimation of the threshold, we are going to run the \textsf{Prune} procedure. This procedure takes as input the set $S_1$ and the threshold $L$, and returns the set $S_2$, which is a subset of $S_1$.
We need the following lemma, which indicates that, after the \textsf{Prune} procedure returns, the $\ell_2$ mass of $x_{\bar{S_2}}$ is comparable to the $\ell_2$ mass of $x_{\mathrm{tail}(k)}$. This lemma is helpful both for finding the relative signs, but also for bounding the error of our approximation of $x$. For definitions of $l_0$ we remind the reader to look at the pseudocode of the \textsf{Prune} algorithm. 

Why do we need the \textsf{Prune} procedure?
Our initial goal is to find the relative signs among elements in the set we have in our hands. Depending on the size of the set, we restrict our attention to the corresponding matrix: the matrix $F_{2^l}$, where $2^l$ is the least positive integer which is not smaller than $|S_2|$. Each such matrix is essentially a sampling procedure: it samples a number of random pairs (edges) from $S_2$. If every entry of $F_{2^l}$ is non-zero with probability $p$, this means that $ { |S_2| \choose 2 }p^2(1-p)^{ |S_2|-2}$ is roughly the probability that a pair in $S_2$ is sampled. Consider now that the size of $S_2$ is very small, let us say constant. Then if $p$ is very small, we need a lot of rows till we sample a pair from $S_2$. If $p$ is large enough, then the expected noise coming into each measurement will be too big and hence most of the time we will not be able to deduce any sign information from that measurement. 

The \textsf{Prune} procedure, hence, keeps only a subset of coordinates of $S_1$ by trying to balance three things: the expected `noise' coming into each measurement (which affects the information we get from that measurement), the time till a pair is sampled (which affects the number of measurements) and the error introduced by throwing out coordinates with lower mass (which affects the approximation guarantee of our scheme). We move on with the following lemma that treats the last of the three aforementioned facts.

\begin{lemma}
\label{lem:massremains}
 With probability $1 - \frac{1}{poly(n)}$, $\|x_{\bar{S_2}}\|_2^2 \leq 1.3 \|x_{\mathrm{tail}(k)} \|_2^2$. \\

\end{lemma}

\begin{proof}

It holds that 

\[  \|x_{\bar{S}_2} \|_2^2 \leq 1.2 \|x_{\mathrm{tail}(k)}\|_2^2 + \sum_{l_0=1}^{\lceil \log(5k) \rceil} 2^l \frac{k L}{C_0 2^l (\log(5k) - l_0 +2)^2} \leq \]

\[\|x_{\mathrm{tail}(k)}\|_2^2 + \frac{kL}{C_0} \sum_{i=1}^{\infty} \frac{1}{i^2} \leq  1.3 \|x_{\mathrm{tail}(k)}\|_2^2 , \]

if $C_0$ is sufficiently large.

\end{proof}

Now, let $2^{l^*}$ be the smallest power of $2$ that exceeds $|S_2|$. We focus on the matrix $F_{2^{l^*}}$. We initialize an empty graph $G$ with vertex $S_2$. We keep each row $q$ whose support has intersection exactly $2$ with $S_2$. Then, if $\sigma_{q,u} =\sigma_{q,v}$ and $| y_q - | |\hat{x}_u| +  |\hat{x}_v||| <  | y_q - | |\hat{x}_u| - |\hat{x}_v|| |$ we add an edge $\{u,v\}$ to $G$. Also we add an edge $\{u,v\}$ to $G$ if $\sigma_{q,u} \neq \sigma_{q,v}$ and $| y_q - | |\hat{x}_u| +  |\hat{x}_v||| >  | y_q - | |\hat{x}_u| - |\hat{x}_v|| |$. Then, we run the \textsf{Degree Profiling} algorithm on $G$ and get a clustering of $S_2$. We arbitrarily then assign positive signs to coordinates belonging to one cluster returned by the algorithm and negative signs to coordinates belonging to the other cluster. We then output the vector supported on $S_2$ with estimates of coordinates we obtained from \textsf{CountSketch} and the signs implied by the aforementioned clustering algorithm.

\begin{lemma}
\label{lem:findsigns}
With probability at least $1 - o_{|S_2|}(1)$, the algorithm will find correctly the relative signs of all elements in the set $S_2$.
\end{lemma}

\begin{proof}
We remind that  $2^{l^*}$ is the smallest power of two that exceeds $S_2$. 

The probability that the intersection of the support of a row of $F_{2^{l^*}}$ and $S_2$ has size exactly $2$ equals ${ |S_2| \choose 2} (\frac{1}{C_0 2^{l^*} (\log (5k) - l^*+2)^2})^2 (1 - \frac{1}{C_0 \cdot 2^{l^*} (\log (5k) - l^* +2)^2})^{|S_2| -2}   = \Omega((\frac{1}{(\log (5k) -l^*+2)^2})^2)$. Thus, after considering  $\Oh(l^* \cdot 2^{l^*}  (\log (5k) - l^* +2)^4  )$ rows we will have seen $\Oh(l^* \cdot 2^{l^*} )$ rows of $F_{2^{l^*}}$ that have intersection exactly $2$ with $S_2$, with probability at least $1- \Oh(\frac{1}{2^{l^*}})$. This means that we will have sampled $\Omega( l^* \cdot 2^{l^*} )$  pairs from $S_2$. Moreover, each pair corresponds to a specific row of $F_{2^{l^*}}$. Fix a pair $\{u,v\}$ and let $q$ be the measurement corresponding to that pair. If $C_0$ is large enough, with constant probability the squared $\ell_2$ mass of coordinates outside of $S_2$ that participates in $y_q$ is at most $\frac{ 4}{C_0 2^{l^*} (\log(5k) - l^* +2)^2} \|x_{\mathrm{tail}(10k)}\|_2^2$ with probability at least $\frac{2}{3}$. Let us focus in the case that $\sigma_{q,u} =\sigma_{q,v}$. If $x_u, x_v$ are of the same sign then
\[| y_q - | |\hat{x}_u| +  |\hat{x}_v||| <  | y_q - | |\hat{x}_u| - |\hat{x}_v|| |\] with probability $\frac{2}{3}$. 

An analogous claim can be made when the signs $\sigma_{q,u}, \sigma_{q,v} $ of $x_u,x_v$ are opposite. This implies that the query to each pair will be correct with probability $\frac{2}{3}$. This indicates the right relative sign between the two coordinates with probability $\frac{2}{3}$. Moreover, every edge we sample is uniformly at random. Thus, we have $C_e l^* 2^{l^*}$ edges, for some big constant $C_e$.

 Let $RED$ be the cluster of coordinates $x_i$ such that $i \in S_2$ and $x_i > 0$. Let $BLUE$ be the set of coordinates $x_i$ such that $i \in S_2$ and $x_i < 0$. Observe now that there is an edge between the two clusters with probability $\frac{1}{3} \frac{C_e \cdot l^*}{ 2^{^*}}$, while an edge between elements of the same cluster exists with probability $\frac{2}{3} \frac{C_e \cdot l^*}{2^{l^*}}$. Since in the beginning of the algorithm, we applied a diagonal random sign matrix, this implies that the signs of $x$ are random, which in turn means that the colors of the clusters are random. For this reason, we are exactly in the setting of the stochastic block model. We can now feed the edges we have to the \textsf{Degree Profiling} algorithm. Then Theorem $4$ guarantees that we can find all relative signs between coordinates in $S_2$ and we are done.

\end{proof}

Equipped with Lemma ~\ref{lem:findsigns}, the proof of $Theorem~2$ follows by the following series of inequalities:
 
\[ \mathrm{min}\{ \| \hat{x} - x\|_2^2, \|\hat{x} +x\|_2^2  \}  = \| \hat{x}_{S_2} - x_{S_2} \|_2^2 + \| \hat{x}_{ [n] - S_2} - x_{[n] - S_2} \|_2^2 \]
\[ \leq  5k \frac{1}{10k} \|x_{\mathrm{tail}(k)}\|_2^2 + \|x_{\bar{S_2}}\|_2^2 \leq 1.8 \|x_{\mathrm{tail}(k)}\|_2^2, \] 
where the bound on $\|x_{\bar{S_2}}\|_2^2$ follows by Lemma ~\ref{lem:massremains} and, moreover, we know that there are at most $5k$ coordinates in $S_2$. \\

\subsection{From constant to $o_n(1)$ failure probability}

In this section we show how the algorithm described can be modified, so that we get the guarantee of Theorem 2.
One can observe that the parts of our decoding algorithm that use matrices $A,B,E$ succeed with probability $1-\frac{1}{poly(n)}$. The reason we have constant probability is because of the matrix $F$. More specifically, since the \textsf{Degree Profiling} algorithm succeeds with probability $1-o_N(1)$ for a graph of size $N$, when we use the matrix $F_{2^j}$, we have success probability $1 - o_{2^j}(1)$. However, for small $j$, this might be constant.

However, there is a way to circumvent this. Fix an element $i^{*} \in S_2$ and let us say that $x_{i^*}$ has positive value. Our goal is to find the relative signs of elements in $S_2$ with respect to $x_{i^*}$. First, we will make a scheme that has $o_k(1)$ failure probability. Then, repeating the scheme $\Oh( \log_k n)$ times and for each coordinate $i$ taking the majority vote for the relative sign of $x_i$ with $x_{i^*}$, we drive the failure probability down to $o_n(1)$. The number of measurements now becomes $\Oh(k \log n + k\log k  \log_k n ) = \Oh(k \log n )$.

What remains is to prove that we can modify the scheme from the previous subsection so that it has $o_k(1)$ failure probability. To do so, we focus on matrices $F_2,\ldots, F_{2^M}$, where $2^M$ is the first power of $2$ that exceeds $\sqrt{k}$. Now, for each such matrix $F_{2^l}$, we draw $\Theta( \log k)$ matrices $F_{2^l}^{(1)},\ldots, F_{2^l}^{\Oh( \log k)}$ from the same distribution as $F_{2^l}$, and replace $F_{2^l}$ with the vertical concatenation of these matrices. When $2^l$ is the smallest power of $2$ that exceeds $|S|$, we look at all these $\Oh( \log k )$ matrices; each one gives a set of relative signs for the set $\{x_i\}_{i \in S}$ with respect to $x_{i^*}$. Taking for each coordinate the majority vote for its relative sign, we can drive the failure probability down to $o_k(1)$.
Note that if $|S| > \sqrt{k}$, the failure probability is at most $o_k(1)$.

Since the scheme we suggested in the previous paragraph fails with probability at most $o_k(1)$, repeating $\Oh(\log_k n)$ times and taking the majority vote for its coordinate in $S_2$, gives failure probability $o_n(1)$, as we mentioned in the second paragraph. 

\subparagraph*{Acknowledgements}

The author thanks Jelani Nelson for numerous helpful discussions.

\bibliographystyle{alpha}
\newcommand{\etalchar}[1]{$^{#1}$}

\appendix

\section{Appendix}

Here we prove the following theorem.

\begin{theorem}
There exists a deterministic construction of a matrix $\Phi \in \mathbb{C}^{(6k-2) \times n}$ such that the following holds: Given $y= |\Phi x|$, where $x \in \mathbb{C}^n$ and $\|x\|_0 =k$, we can find a vector $\hat{x}$ such that there exists $\phi \in [0, 2\pi]$ which satisfies $ e^{i \phi} \hat{x} = x$. The running time is $O(k^2)$.
\end{theorem}

We need the following result, called also Prony's method, which is developed in the context of linear compressed sensing, where we also have access to the signs of the measurements. Let $\mathcal{F} \in \mathbb{C}^n$ be the Discrete Fourier matrix. Prony's method indicates how to reconstruct any $k$-sparse vector $x \in \mathbb{C}^n$ given access only to its first $2k$ Fourier coefficients. For a matrix $A$ and an integer $a$, $A_a$ is the matrix that consists of the first $a$ rows of $A$.

\begin{lemma}[\cite{moitra2014algorithmic}]
Suppose $x \in \mathbb{C}^n$ with $\|x\|_0\leq k$. Given $y= \mathcal{F}_{2k} x$, we can find $x$ in time $O(k^2 \cdot)$.\end{lemma}

Note that there is no discussion in \cite{moitra2014algorithmic} for the running time; the number of rows and the correctness is only verified there. More relevantly, the same algorithm appears in in \cite[Theorem $4.4$]{ghazi2013sample} with running time $O(k^2 + k (\log \log n)^c)$ for $k = O( \log n)$, and $c$ an absolute constant. The reason for the constraint in the sparsity $k$ and the slight dependence on $\log\log n$ is the usage of algorithms for finding roots of polynomials and the requirement for precision up to $\Theta(\log n \cdot \log \log n)$ bits. We remark that in principle Prony's method procceeds in three stages: the first stage is the syndrome decoding of Reed-Solomon codes which takes time $O(k^2)$ \cite{va}, the second is an algorithm for finding the roots of a degree $k$ polynomial, and the third is the solution of a linear system where corresponding matrix is Vandermorde, and can be done in time $O(k^2)$ \cite{zippel1990interpolating}. We note that the the second and the first step can be subsituted alltogether by the methods in \cite{hankel} or \cite{qiao2002exponential} which give $O(k^2)$ time.

\begin{lemma}
Let $x,a$ be two complex numbers. If we know $|x|,a, |x+a|, |x+a\sqrt{-1}|$, then we can find $x$.
\end{lemma}

\begin{proof}
Let $x= |x| e^{\sqrt{-1} \phi}$ and $a = |a| e^{\sqrt{-1}\psi}$.
It holds that $|x+a|^2 = |x|^2 + |a|^2 + x \bar{a} + \bar{x} a $ and hence 

$ |x||a| e^{\sqrt{-1} (\phi - \psi)} + |x| |a| e^{\sqrt{-1} (-\phi + \psi) } = |x+a|^2 - |x|^2 - |a|^2$.

This implies that $e^{\sqrt{-1}(\phi - \psi)} + e^{\sqrt{-1}(-\phi + \psi)} = \frac{|x+a|^2 - |x|^2 - |a|^2}{|x||a|}$. 

Similarly from $|x|,a,|x+a\sqrt{-1}|$ we can learn

\[e^{\sqrt{-1}( \phi - \psi)} - e^{\sqrt{-1}(- \phi + \psi)}.	\]

Given these two values, by a simple substraction we can learn $e^{\sqrt{-1}(\phi - \psi)}$ and thus $\phi$ in terms of $\psi$, and consequently $x$.

\end{proof}

We are now ready to indicate the construction of the matrix that satisfies the guarantees of Theorem 2. Let $\{e_i\}_{i=1}^n$ be the standard basis vectors of $\mathbb{R}^n$. Define $A$ to be the $(2k-1) \times n$ matrix the $a$-th row of which equals $\sum_{i=1}^{a+1} e_i$. Define $B$ to be the $(2k-1) \times n$ matrix the $b$-th row of which equals $\sum_{i=1}^b e_i + \sqrt{-1} e_{b+1}$. Let $C$ be the vertical concatenation of $A,B$ and $I_{2k}$.

Now, set $\Phi = C \mathcal{F}$. We show that our matrix $\Phi$ satisfies the guarantee of Theorem 2.\newline

\begin{proof}
First of all, $\Phi$ has $2k + 2\cdot(2k-1) = 6k-2$ rows.

From $I_{2k} \mathcal{F}$, we get the magnitude of the first $2k$ Fourier Coefficients of $x$, let them be $z_1, \ldots, z_{2k}$. If we now find the relative phases we can then use Prony's method and reconstruct $x$ up to a global phase factor. 

Let $i^*$ be the minimum index such that $|z_{i^*}|$ is non-zero. If there is no such index, clearly $x=0$ and we are done. Otherwise, we set the phase of $z_{i^*}$ to be $0$ and we are going to find the phase of every other coordinate with respect to the phase of $z_{i^*}$. 

For any $j > i^*$, by the $(j-1)$-th row of $A\mathcal{F}$ and $B \mathcal{F}$ we have both $| z_{j} + \sum_{i=1}^{j-1} z_i |$ and $ |\sqrt{-1} \cdot z_j + \sum_{i=1}^{j-1}z_i|$. By considering all $j$ in increasing order, at each time $j$ we know $z_i$ for $i<j$ and hence $\sum_{i=j^\ast}^{j-1}z_i$. Invoking the previous lemma, we can find $z_j$ and hence using Prony's method we can find $\hat{x}$ which satisfies the guarantees of Theorem 2.
\end{proof}

\end{document}